\pdfoutput=1
\documentclass[10pt,twocolumn]{IEEEtran}
\usepackage{balance}
\usepackage[cmex10]{amsmath}
\usepackage{amsfonts,mathrsfs,amsthm,amssymb}
\usepackage{cite}
\usepackage[ruled,longend]{algorithm2e}
\usepackage{graphicx}
\usepackage[font={footnotesize}]{caption}
\usepackage[font={footnotesize}]{subcaption}

\newtheorem{theorem}{Theorem}
\newtheorem{proposition}{Proposition}

\DeclareMathOperator{\toep}{toep}
\DeclareMathOperator{\diagtoep}{d-toep}
\DeclareMathOperator{\trace}{tr}
\DeclareMathOperator{\rank}{rank}
\DeclareMathOperator{\spn}{span}
\DeclareMathOperator{\sign}{sign}
\DeclareMathOperator{\sdp}{SDP}

\DeclareMathOperator{\st}{s.t.\ }
\DeclareMathOperator{\real}{Re}

\DeclareMathOperator{\diag}{diag}
\ifCLASSINFOpdf
\else
\fi
\begin{document}
%
% paper title
% can use linebreaks \\ within to get better formatting as desired
% Do not put math or special symbols in the title.
\title{Distributed Compressed Sensing off the Grid}
%
%
% author names and IEEE memberships
% note positions of commas and nonbreaking spaces ( ~ ) LaTeX will not break
% a structure at a ~ so this keeps an author's name from being broken across
% two lines.
% use \thanks{} to gain access to the first footnote area
% a separate \thanks must be used for each paragraph as LaTeX2e's \thanks
% was not built to handle multiple paragraphs
%

%\author{Michael~Shell,~\IEEEmembership{Member,~IEEE,}
%        John~Doe,~\IEEEmembership{Fellow,~OSA,}
%        and~Jane~Doe,~\IEEEmembership{Life~Fellow,~IEEE}% <-this % stops a space
%\thanks{M. Shell is with the Department
%of Electrical and Computer Engineering, Georgia Institute of Technology, Atlanta,
%GA, 30332 USA e-mail: (see http://www.michaelshell.org/contact.html).}% <-this % stops a space
%\thanks{J. Doe and J. Doe are with Anonymous University.}% <-this % stops a space
%\thanks{Manuscript received April 19, 2005; revised December 27, 2012.}}
\author{Zhenqi~Lu*,~Rendong~Ying,~Sumxin~Jiang,
~Peilin~Liu,~\IEEEmembership{Member,~IEEE,}~and~Wenxian~Yu,~\IEEEmembership{Member,~IEEE}
\thanks{*e-mail: zhenqilu2014@gmail.com}
\thanks{The authors are with the Dept. of Electrical Engineering, Shanghai Jiao Tong University, Shanghai, P.R. China. This work was partially supported by NSFC under grant number 61171171.}}

% note the % following the last \IEEEmembership and also \thanks -
% these prevent an unwanted space from occurring between the last author name
% and the end of the author line. i.e., if you had this:
%
% \author{....lastname \thanks{...} \thanks{...} }
%                     ^------------^------------^----Do not want these spaces!
%
% a space would be appended to the last name and could cause every name on that
% line to be shifted left slightly. This is one of those "LaTeX things". For
% instance, "\textbf{A} \textbf{B}" will typeset as "A B" not "AB". To get
% "AB" then you have to do: "\textbf{A}\textbf{B}"
% \thanks is no different in this regard, so shield the last } of each \thanks
% that ends a line with a % and do not let a space in before the next \thanks.
% Spaces after \IEEEmembership other than the last one are OK (and needed) as
% you are supposed to have spaces between the names. For what it is worth,
% this is a minor point as most people would not even notice if the said evil
% space somehow managed to creep in.

% The paper headers
\markboth{Lu \emph{\lowercase{et al.}}: Distributed Compressed Sensing off the Grid}{TBD}%
% The only time the second header will appear is for the odd numbered pages
% after the title page when using the twoside option.
%
% *** Note that you probably will NOT want to include the author's ***
% *** name in the headers of peer review papers.                   ***
% You can use \ifCLASSOPTIONpeerreview for conditional compilation here if
% you desire.

% If you want to put a publisher's ID mark on the page you can do it like
% this:
%\IEEEpubid{0000--0000/00\$00.00~\copyright~2012 IEEE}
% Remember, if you use this you must call \IEEEpubidadjcol in the second
% column for its text to clear the IEEEpubid mark.

% use for special paper notices
%\IEEEspecialpapernotice{(Invited Paper)}

% make the title area
\maketitle

% As a general rule, do not put math, special symbols or citations
% in the abstract or keywords.
\begin{abstract}
This letter investigates the joint recovery of a frequency-sparse signal ensemble sharing a common frequency-sparse component from the collection of their compressed measurements. Unlike conventional arts in compressed sensing, the frequencies follow an off-the-grid formulation and are continuously valued in $\left\lbrack 0,1 \right\rbrack$. As an extension of atomic norm, the concatenated atomic norm minimization approach is proposed to handle the exact recovery of signals, which is reformulated as a computationally tractable semidefinite program. The optimality of the proposed approach is characterized using a dual certificate. Numerical experiments are performed to illustrate the effectiveness of the proposed approach and its advantage over separate recovery.
\end{abstract}

% Note that keywords are not normally used for peerreview papers.

% For peer review papers, you can put extra information on the cover
% page as needed:
% \ifCLASSOPTIONpeerreview
% \begin{center} \bfseries EDICS Category: 3-BBND \end{center}
% \fi
%
% For peerreview papers, this IEEEtran command inserts a page break and
% creates the second title. It will be ignored for other modes.
\IEEEpeerreviewmaketitle

\section{Introduction}
\IEEEPARstart{C}{ompressed Sensing} (CS) is an emerging theory enabling sub-Nyquist sampling via combination of signal acquisition and signal compression \cite{donoho2006compressed,candes2006robust,candes2008introduction,baraniuk2007compressive} . Despite its remarkable impact on a wide range of signal processing theory and methods, conventional CS developments are constrained to signals with sparse or compressible representations on a pre-defined grid \cite{duarte2013spectral,candes2011compressed,fannjiang2012coherence}. However, in applications including communciation, radar, seismology, localization and remote sensing, signals of interest are usually specified by parameters in a continuous domain \cite{stoica2005spectral,ekanadham2011recovery,malioutov2005sparse,parrish2009improved,meng2011collaborative,mishali2011xampling,barilan2014subnyquist,fang2014super}. Performance degradation due to \emph{basis mismatch} between real parameters and pre-defined grid is studied and addressed \cite{chi2011sensitivity,scharf2011sensitivity}, and many approaches have been proposed to mitigate this effect \cite{mishra2014spectral,fyhn2013spectral,nichols2014reducing}.
\newline\indent Most recently, a group of works has concentrated on the obviation of the basis mismatch conundrum. It has been shown that a frequency-sparse signal can be successfully recovered from its consecutive sub-Nyquist samples using total-variation minimization \cite{rudin1987real}, which can be solved via semidefinite program (SDP), where only a minimum separation between spectral spikes is required \cite{candes2013towards}. The usage of atomic norm \cite{chandrasekaran2007convex} extends this work to the random sampling regime, and reaches improved trade-off between minimum spectral separation and number of observations required \cite{tang2013compressed}. This framework has been further extended to cases including multiple measurement vectors \cite{chi2014joint,liao2014music}, two-dimensional frequencies \cite{chi2013compressive}, direction-of-arrival estimation \cite{tan2013direction,fang2014super}, spectrum estimation with block prior information \cite{mishra2014super}, etc. Another recent approach is to apply matrix pencil \cite{hua1992estimating} to CS, and reformulate the problem as structured matrix completion \cite{chen2013spectral}. In addition, inspired by the idea of model selection \cite{nadler2011model}, the recovery problem is resolved as a parametric estimation problem via order selection \cite{lu2014spectral,nielsen2014joint}, which can be solved efficiently using greedy methods.
\newline\indent In this letter, we address the problem of simultaneously recovering a \emph{joint frequency-sparse} (JFS) signal ensemble sharing a common frequency-sparse component, with frequencies continuously valued in $\left\lbrack 0,1 \right\rbrack$. This common/innovation joint sparsity model is shown to significantly reduce the number of measurements in conventional distributed CS framework by utilizing common information shared in multiple signals \cite{baron2009distributed}. Our main contribution is to develop the continuous counterpart of the joint sparsity model, and propose the \emph{concatenated atomic norm} (CA-norm) for the description of joint frequency sparsity, of which the minimization can be solved via SDP. We also characterize a dual certificate for the optimality of the proposed optimization problem. Numerical results are given to illustrate the effectiveness of our approach and its advantage over separate recovery, which indicate a significant reduction in the number of measurements per signal required for successful recovery. Empirical observations also show improved performance for ensemble involving a large number of signals, implying the promise of application to large-scale sensor systems including MIMO communication, sensor array, multi-antenna, radar array, etc., where signals sensed are affected by structured global (common) factors and structured local (innovation) factors combined.

\section{Joint Frequency-Sparse Signal Ensemble}
Let $\Lambda = \left\lbrace 1,2,\ldots,J \right\rbrace$ denote the set of indices for the $J$ signals in the ensemble. Denote the \emph{signals} in the ensemble by $x_j\in\mathbb{C}^n$, and assume that each $x_j$ is frequency-sparse. The signal ensemble is denoted by $X = \lbrack x_1^{\ast},\ldots,x_J^{\ast} \rbrack^{\ast}$. The superscript $^{\ast}$ means Hermitian transpose. Each signal $x_j$ is sensed using a different sensing matrix $\Phi_j\in\mathbb{C}^{m_j\times n}$, and the corresponding measurement is denoted by $y_j=\Phi_jx_j\in\mathbb{C}^{m_j}$. Define $\Phi = \diag\left(\Phi_1,\ldots,\Phi_J\right)$. In the JFS setting, we additionally assume that each signal is generated as a combination of two frequency-sparse components: $\left(i\right)$ a common component $z_c$, which is present in all signals, and $\left(ii\right)$ an innovation component $z_j$, which is unique to each signal. The component ensemble is denoted by $Z = \left\lbrack z_c^{\ast},z_1^{\ast},\ldots,z_J^{\ast} \right\rbrack^{\ast}$. These combine additively, giving $x_j = z_c + z_j,j\in\Lambda$.
\newline\indent The frequency-sparse property of the components implies that these can be expressed as
\begin{IEEEeqnarray*}{rCl}
    z_c &=& \sum_{k=1}^{s_c} \left\lvert c_{c,k} \right\rvert a\left(f_{c,k},\phi_{c,k}\right) = \sum_{k=1}^{s_c} c_{c,k} a\left(f_{c,k}\right),\\
    z_j &=& \sum_{k=1}^{s_j} \left\lvert c_{j,k} \right\rvert a\left(f_{j,k},\phi_{j,k}\right) = \sum_{k=1}^{s_j} c_{j,k} a\left(f_{j,k}\right),
\end{IEEEeqnarray*}
where the atoms $a(f,\phi) = e^{i\phi}a(f)\in\mathbb{C}^n$, $f\in\left\lbrack 0,1 \right\rbrack$, $\phi\in\lbrack 0,2\pi )$ are defined as
\begin{equation}
    \left\lbrack a\left(f,\phi\right) \right\rbrack_t = e^{i\left(2\pi ft+\phi\right)},t\in L = \left\lbrace 0,\ldots,n-1 \right\rbrace.
\end{equation}
The sets of frequencies are defined as $\Omega_c = \left\lbrace f_{c,k} \right\rbrace_{k=1}^{s_c}$ and $\Omega_j = \left\lbrace f_{j,k} \right\rbrace_{k=1}^{s_j}$, and frequencies $f_{c,k},f_{j,k}$ are continuously valued in $\left\lbrack 0,1 \right\rbrack$.

\section{Concatenated Atomic Norm and Semidefinite Program Characterization}
\subsection{Concatenated Atomic Norm}
Define the atom set as
\begin{equation}
    \mathcal{A} = \left\lbrace a\left( f,\phi \right):f\in\left\lbrack 0,1 \right\rbrack, \phi\in\left\lbrack 0,2\pi \right) \right\rbrace,
\end{equation}
and the ''$\ell_0$-norm'' type atomic norm \cite{tang2013compressed} is defined as
\begin{equation}
    \left\lVert x \right\rVert_{\mathcal{A},0} = \inf\bigg\lbrace s : x = \sum_{k=1}^s \left\lvert c_k\right\rvert a\left( f_k,\phi_k \right) \bigg\rbrace,
\end{equation}
and its convex relaxation, the atomic norm \cite{chandrasekaran2007convex}, is defined as
\begin{equation}\label{definition:atomic_norm}
    \left\lVert x \right\rVert_{\mathcal{A}} = \inf\bigg\lbrace \sum_k \left\lvert c_k\right\rvert : x = \sum_k \left\lvert c_k\right\rvert a\left( f_k,\phi_k \right) \bigg\rbrace.
\end{equation}
To develop a norm description of the joint sparsity, we extend the atomic norm to the JFS setting and give the definition of CA-norm. The ''$\ell_0$-norm'' type CA-norm is defined as
\begin{equation}
    \left\lVert X \right\rVert_{\mathcal{CA},0} = \inf\Big\lbrace\left\lVert z_c\right\rVert_{\mathcal{A},0} + \sum_{j\in\Lambda}\left\lVert z_j\right\rVert_{\mathcal{A},0} :
    z_c + z_j = x_j,j\in\Lambda \Big\rbrace,
\end{equation}
and thus our goal becomes the minimization of $\left\lVert X \right\rVert_{\mathcal{CA},0}$ satisfying the measurement \emph{a-priori}
\begin{equation}
    \min_X \left\lVert X \right\rVert_{\mathcal{CA},0}\ \st y_j = \Phi_j x_j,j\in\Lambda,
\end{equation}
which can be shown to be equivalent to the following rank minimization problem using approach similar to Theorem \ref{theorem:semidefinite}
\begin{IEEEeqnarray*}{rCl}
    &\min_{\boldsymbol{u},Z,t}& \frac{1}{2n} \Big( \rank\left(\toep(u_c)\right) + \sum_{j\in\Lambda} \rank\left(\toep(u_j)\right) \Big)\\
    &\st&
    \begin{bmatrix}
        \diagtoep(\boldsymbol{u}) & Z \\
        Z^{\ast} & t
    \end{bmatrix} \succeq 0,
    y_j = \Phi_j (z_c + z_j),j\in\Lambda,\IEEEyesnumber
\end{IEEEeqnarray*}
where $\diagtoep(\boldsymbol{u})$ is the block diagonal matrix
\begin{equation*}
    \diag\big( \toep(u_c),\toep(u_1),\ldots,\toep(u_J) \big)
\end{equation*}
composed of toeplitz matrices generated from complex vectors $\boldsymbol{u} = \lbrace u_c,u_j,j\in\Lambda \rbrace$. Due to the NP-hard nature of rank minimization problem, solving the ''$\ell_0$-norm'' type CA-norm minimization would become computationally intractable. An alternative approach is to consider its convex relaxation, CA-norm, defined as
\begin{equation}\label{concatenated:atomic_norm}
    \left\lVert X\right\rVert_{\mathcal{CA}} = \inf\Big\lbrace \left\lVert z_c\right\rVert_{\mathcal{A}} + \sum_{j\in\Lambda}\left\lVert z_j\right\rVert_{\mathcal{A}} : z_c + z_j = x_j,j\in\Lambda \Big\rbrace.
\end{equation}
The atomic norm defined for single vector in \eqref{definition:atomic_norm} is actually a special case of CA-norm for $J=1$. In this work, we propose to solve the following CA-norm minimization problem to achieve accurate recovery of off-the-grid joint frequency-sparse signal
\begin{equation}\label{CA:minimization}
    \min_X \left\lVert X \right\rVert_{\mathcal{CA}}\ \st y_j = \Phi_j x_j,j\in\Lambda.
\end{equation}
\subsection{Semidefinite Program Characterization}
In this section, we prove the equivalence between CA-norm and SDP, and thus propose the computationally tractable SDP solution to the CA-norm minimization problem \eqref{CA:minimization}.
\begin{theorem}\label{theorem:semidefinite}
    For $x_j = z_c + z_j\in\mathbb{C}^n,j\in\Lambda$,
\begin{IEEEeqnarray*}{rCl}\label{semidefinite:program:unconstrained}
    \left\lVert X \right\rVert_{\mathcal{CA}} = &\inf& \bigg\lbrace \frac{1}{2n} \Big( \trace\left(\toep(u_c)\right) + \sum_{j\in\Lambda} \trace\left(\toep(u_j)\right) \Big) + \frac{1}{2}t : \\
    &&\begin{bmatrix}
        \diagtoep(\boldsymbol{u}) & Z \\
        Z^{\ast} & t
    \end{bmatrix} \succeq 0 \bigg\rbrace\IEEEyesnumber
\end{IEEEeqnarray*}
\end{theorem}
\begin{proof}
First we define $I_c,I_j,j\in\Lambda$ the $J+1$ submatrices with dimension $(J+1)N\times N$ of the identity matrix $I$ in $\mathbb{R}^{(J+1)N\times(J+1)N}$, shown as $I = \lbrack I_c,I_1,\ldots,I_J \rbrack$. Denote the term on the right side of \eqref{semidefinite:program:unconstrained} by $\sdp\left( X \right)$. For any $z_c$ and $z_j$ satisfying $z_c + z_j = x_j,j\in\Lambda$, suppose
\begin{equation*}
    z_c = \sum_{k=1}^{s_c} \left\lvert c_{c,k} \right\rvert a\left( f_{c,k},\phi_{c,k} \right),z_j = \sum_{k=1}^{s_j} \left\lvert c_{j,k} \right\rvert a\left( f_{j,k},\phi_{j,k} \right).
\end{equation*}
Defining
\begin{equation*}
    u_c = \sum_{k=1}^{s_c} \left\lvert c_{c,k} \right\rvert a\left( f_{c,k} \right),u_j = \sum_{k=1}^{s_j} \left\lvert c_{j,k} \right\rvert a\left( f_{j,k} \right),
\end{equation*}
and $t = \sum_{k=1}^{s_c} \left\lvert c_{c,k} \right\rvert + \sum_{j\in\Lambda}\sum_{k=1}^{s_j} \left\lvert c_{j,k} \right\rvert$ yields
\begin{IEEEeqnarray*}{rCl}
    \toep \left( u_c \right) &=& \sum_{k=1}^{s_c} \left\lvert c_{c,k} \right\rvert a\left( f_{c,k},\phi_{c,k} \right)a\left( f_{c,k},\phi_{c,k} \right)^{\ast},\\
    \toep \left( u_j \right) &=& \sum_{k=1}^{s_j} \left\lvert c_{j,k} \right\rvert a\left( f_{j,k},\phi_{j,k} \right)a\left( f_{j,k},\phi_{j,k} \right)^{\ast},
\end{IEEEeqnarray*}
and thus
\begin{IEEEeqnarray*}{rCl}
    &&\begin{bmatrix}
        \diagtoep(\boldsymbol{u}) & Z \\
        Z^{\ast} & t
    \end{bmatrix}\\
    &=&
    \sum_{k=1}^{s_c}\left\lvert c_{c,k}\right\rvert
    \begin{bmatrix}
        I_ca\left(f_{c,k},\phi_{c,k}\right)\\1
    \end{bmatrix}
    \begin{bmatrix}
        I_ca\left(f_{c,k},\phi_{c,k}\right)\\1
    \end{bmatrix}^{\ast}\\
    &+&
    \sum_{j\in\Lambda} \sum_{k=1}^{s_j} \left\lvert c_{j,k}\right\rvert
    \begin{bmatrix}
        I_ja\left(f_{j,k},\phi_{j,k}\right)\\1
    \end{bmatrix}
    \begin{bmatrix}
        I_ja\left(f_{j,k},\phi_{j,k}\right)\\1
    \end{bmatrix}^{\ast}
\end{IEEEeqnarray*}
is positive semidefinite. It follows that
\begin{IEEEeqnarray*}{rCl}
    &\frac{1}{n}\trace\left(\toep\left(u_c\right)\right) = \sum_{k=1}^{s_c}\left\lvert c_{c,k}\right\rvert,\frac{1}{n}\trace\left(\toep\left(u_j\right)\right) = \sum_{k=1}^{s_j}\left\lvert c_{j,k}\right\rvert,\\
    &t = \frac{1}{n} \Big( \trace\left( \toep\left( u_c \right) \right)+\sum_{j\in\Lambda}\trace\left( \toep\left( u_j \right) \right) \Big).
\end{IEEEeqnarray*}
and thus $\sum_k\left\lvert c_{c,k}\right\rvert + \sum_{j\in\Lambda} \sum_k\left\lvert c_{j,k}\right\rvert \geq \sdp (X)$. Since it holds for any $z_c$ and $z_j$ satisfying $z_c + z_j = x_j$, it follows that $\left\lVert z_c \right\rVert_{\mathcal{A}} + \sum_{j\in\Lambda} \left\lVert z_j \right\rVert_{\mathcal{A}}\geq \sdp (X)$, implying that $\left\lVert X \right\rVert_{\mathcal{CA}} \geq \sdp (X)$.
\newline\indent Conversely, suppose for some $z_c,z_j$ and $u_c,u_j$ satisfying
\begin{equation}\label{semidefinite:inequation}
    \begin{bmatrix}
        \diagtoep(\boldsymbol{u}) & Z \\
        Z^{\ast} & t
    \end{bmatrix} \succeq 0,
\end{equation}
form the Vandermonde decomposition $\toep ( u_c ) = V_c D_c V_c^{\ast}$, $\toep ( u_j ) = V_j D_j V_j^{\ast}$ \cite{caratheodory1911zusammenhang}, where $D_c,D_j$ are positive definite diagonal matrices, and thus $\frac{1}{n}\trace(\toep(u_c)) = \trace(D_c)$ and $\frac{1}{n}\trace(\toep(u_j)) = \trace(D_j)$. It follows that $z_c\in\spn(V_c),z_j\in\spn(V_j)$, hence $z_c = V_c\omega_c,z_j = V_j\omega_j$. The full rank property of $V_c$ and $V_j$ implies that there exist vectors $q_c$ and $q_j$ satisfying $V_c^{\ast}q_c=\sign(\omega_c)$ and $V_j^{\ast}q_j=\sign(\omega_j)$. Define $\boldsymbol{V} = \diag(V_c,V_1,\ldots,V_J)$, $\boldsymbol{D} = \diag(D_c,D_1,\ldots,D_J)$, $\boldsymbol{\omega} = \lbrack \omega_c^{\ast},\omega_1^{\ast},\ldots,\omega_J^{\ast} \rbrack^{\ast}$, and $\boldsymbol{q} = \lbrack q_c^{\ast},q_1^{\ast},\ldots,q_J^{\ast} \rbrack^{\ast}$.
\newline\indent The application of Schur Complement Lemma \cite{boyd2004convex} yields
\begin{equation*}
    \diagtoep (\boldsymbol{u}) - \frac{1}{t} ZZ^{\ast} \succeq 0.
\end{equation*}
Further performing the Vandermonde decomposition, we have
\begin{equation*}
    \boldsymbol{V}\boldsymbol{D}\boldsymbol{V}^{\ast} - \frac{1}{t} \boldsymbol{V}\boldsymbol{\omega}\boldsymbol{\omega}^{\ast}\boldsymbol{V}^{\ast} \succeq 0
\end{equation*}
and thus
\begin{IEEEeqnarray*}{rCl}
    &&\trace(D_c) + \sum_{j\in\Lambda}\trace(D_j) = \boldsymbol{q}^{\ast}\boldsymbol{V}\boldsymbol{D}\boldsymbol{V}^{\ast}\boldsymbol{q}\\
    &\geq& \frac{1}{t} \boldsymbol{q}^{\ast}\boldsymbol{V}\boldsymbol{\omega}\boldsymbol{\omega}^{\ast} \boldsymbol{V}^{\ast}\boldsymbol{q}
    = \frac{1}{t} \Big( \left\lVert \omega_c \right\rVert_1 + \sum_{j\in\Lambda} \left\lVert \omega_j \right\rVert_1 \Big)^2,\IEEEyesnumber
\end{IEEEeqnarray*}
implying that
\begin{equation}
    t\Big(\trace(D_c) + \sum_{j\in\Lambda}\trace(D_j)\Big) \geq \Big( \left\lVert \omega_c \right\rVert_1 + \sum_{j\in\Lambda} \left\lVert \omega_j \right\rVert_1 \Big)^2.
\end{equation}
By the arithmetic geometric mean inequality,
\begin{IEEEeqnarray*}{rCl}
    \sdp(X) &=& \frac{1}{2n} \Big( \trace(\toep(u_c)) + \sum_{j\in\Lambda} \trace(\toep(u_j)) \Big) + \frac{1}{2} t\\
    &=& \frac{1}{2} \Big( \trace(D_c) + \sum_{j\in\Lambda} \trace(D_j) \Big) + \frac{1}{2} t\\
    &\geq& \bigg( t\Big(\trace(D_c) + \sum_{j\in\Lambda}\trace(D_j)\Big) \bigg)^{\frac{1}{2}}\\
    &\geq& \left\lVert \omega_c \right\rVert_1 + \sum_{j\in\Lambda} \left\lVert \omega_j \right\rVert_1 \geq \left\lVert X \right\rVert_{\mathcal{CA}},\IEEEyesnumber
\end{IEEEeqnarray*}
which completes the proof.
\end{proof}
With Theorem \ref{theorem:semidefinite}, \eqref{CA:minimization} is reformulated as the following computationally tractable SDP
\begin{IEEEeqnarray*}{rCl}\label{semidefinite:program}
    &\min_{\boldsymbol{u},Z,t}& \frac{1}{2n} \Big( \trace\left(\toep(u_c)\right) + \sum_{j\in\Lambda} \trace\left(\toep(u_j)\right) \Big) + \frac{1}{2}t\\
    &\st&
    \begin{bmatrix}
        \diagtoep(\boldsymbol{u}) & Z \\
        Z^{\ast} & t
    \end{bmatrix} \succeq 0,
    y_j = \Phi_j (z_c + z_j),j\in\Lambda.\IEEEyesnumber
\end{IEEEeqnarray*}

\subsection{Dual Certificate}
\begin{figure}[tb!]
\centering
\begin{subfigure}{0.158\textwidth}
    \includegraphics[width=\textwidth]{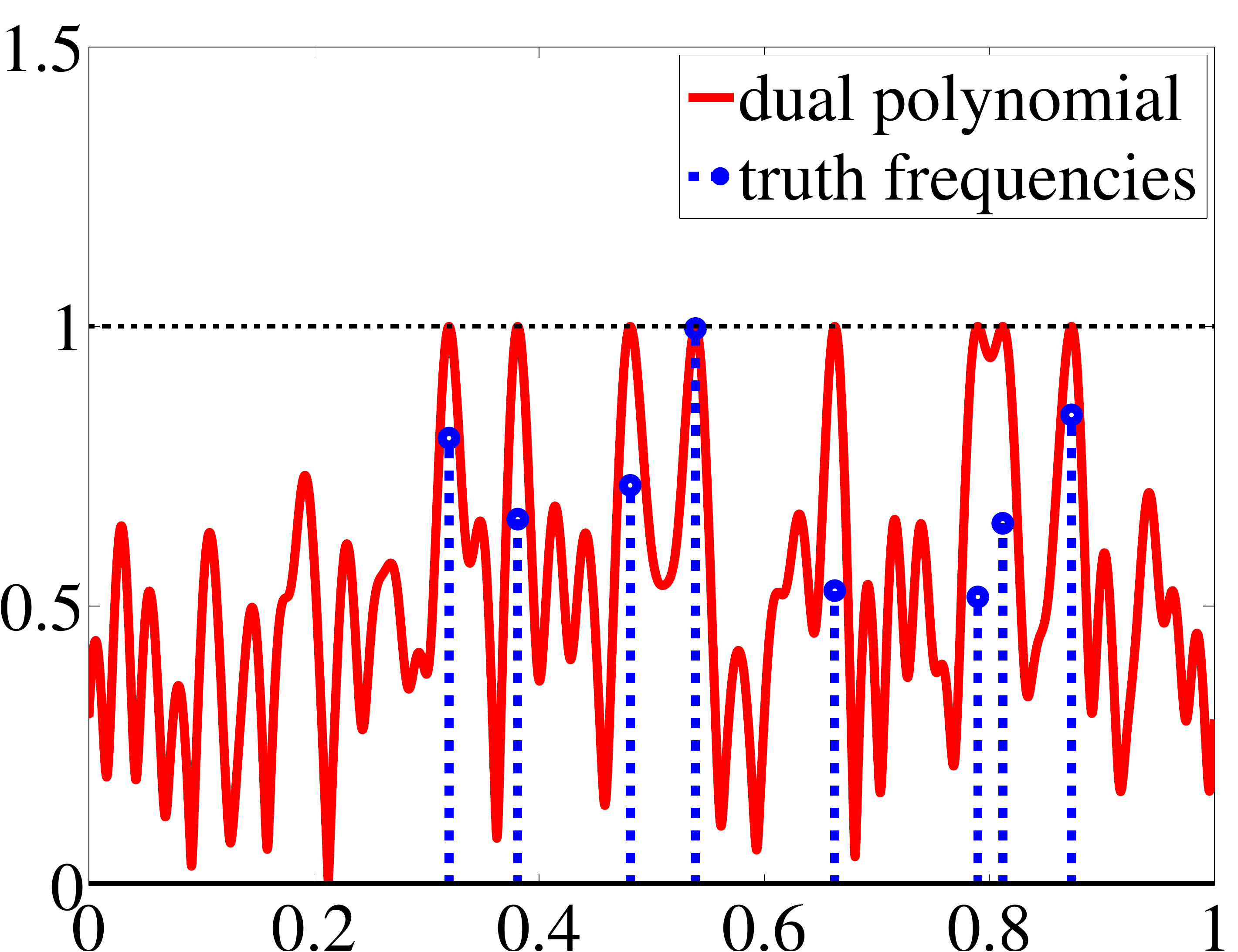}
    \caption{Common frequency}
\end{subfigure}
\begin{subfigure}{0.158\textwidth}
    \includegraphics[width=\textwidth]{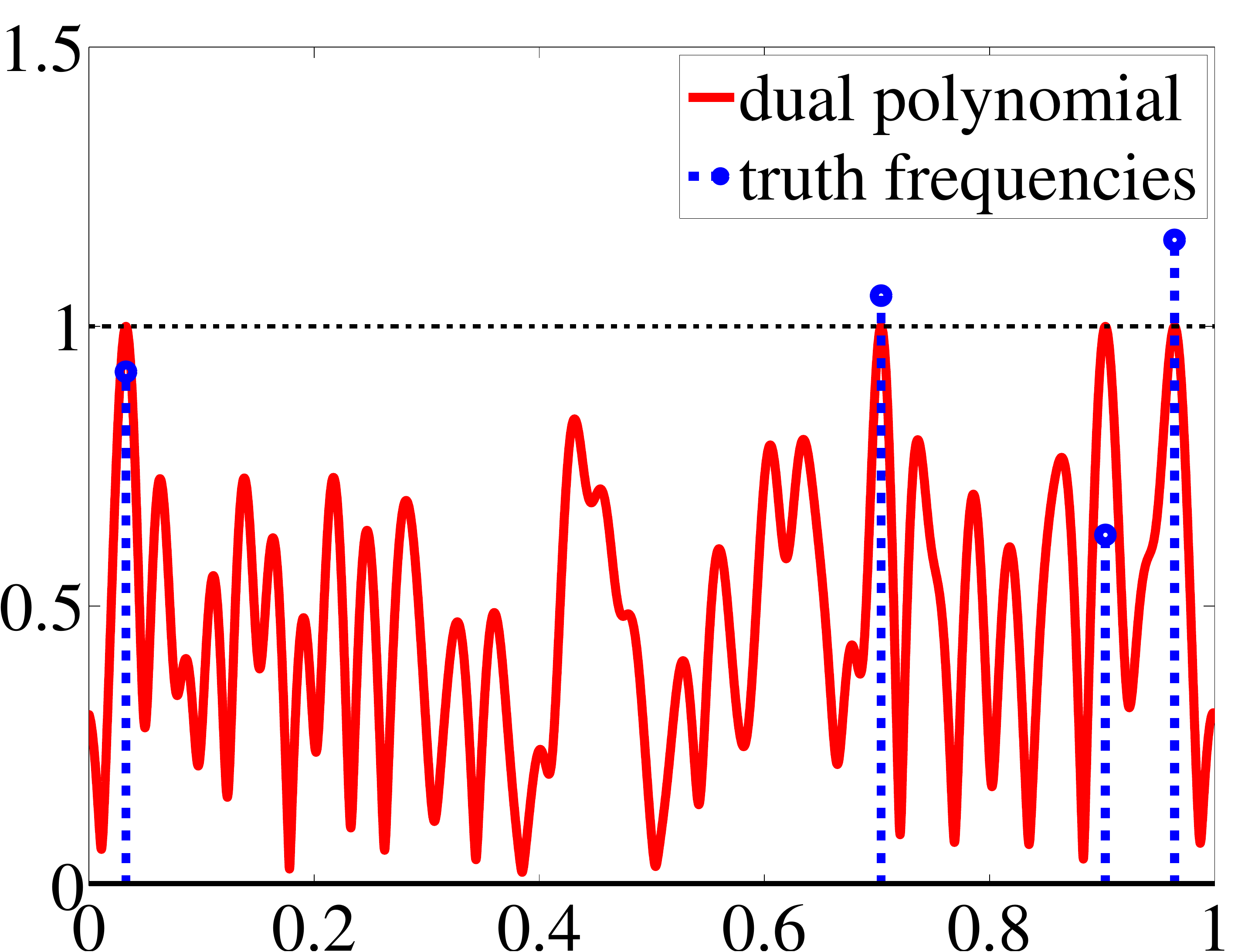}
    \caption{Innovation frequency}
\end{subfigure}
\begin{subfigure}{0.158\textwidth}
    \includegraphics[width=\textwidth]{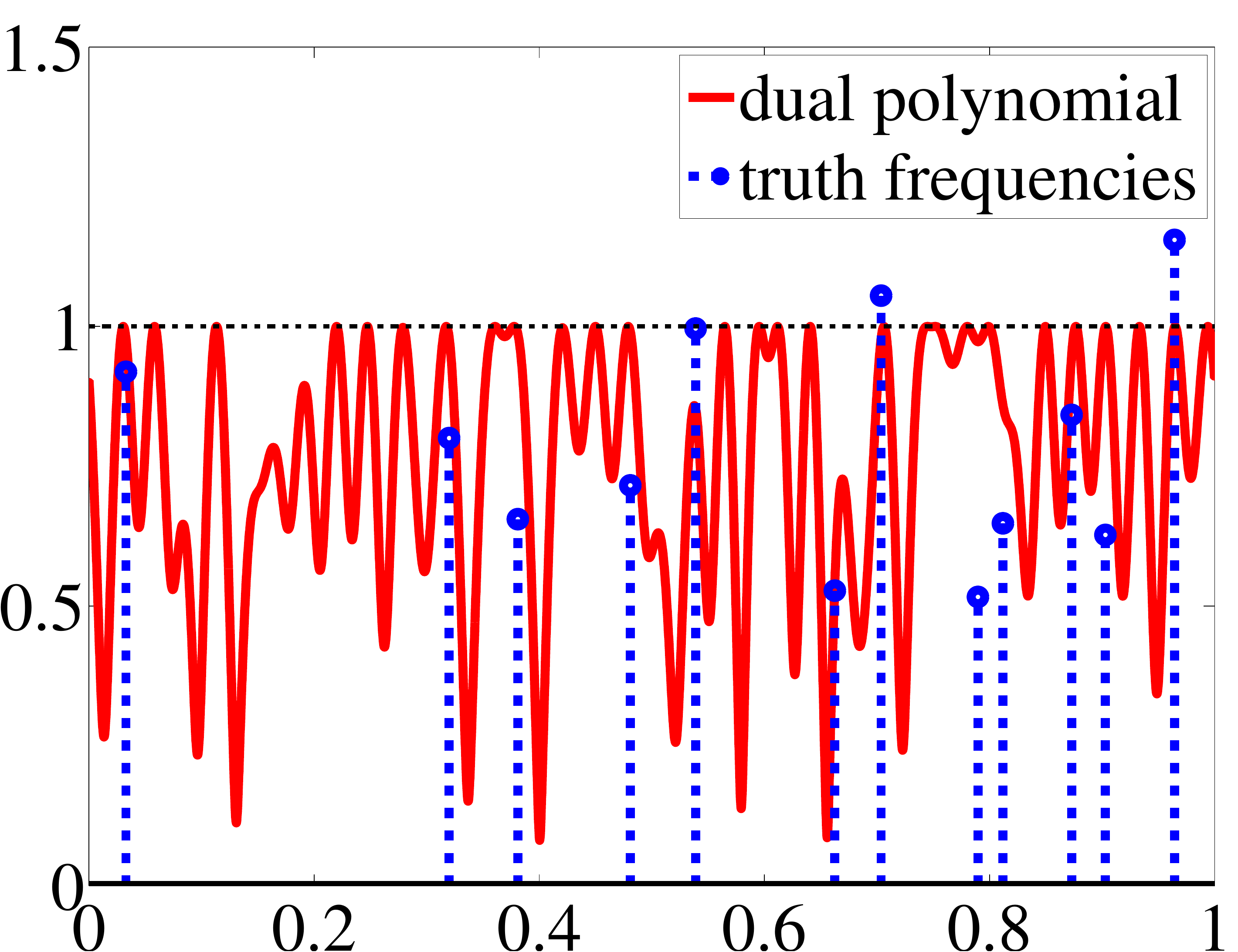}
    \caption{Separate recovery}
\end{subfigure}
\caption{Frequency localization from dual polynomial ($J=32$, $n=40$, $m_j=20,\forall j$)}
\label{figure:frequency:localization}
\end{figure}
%\begin{figure}[tb!]
%\centering
%\begin{subfigure}{0.32\textwidth}
%    \includegraphics[width=\textwidth]{dsdp_comm_freq_local.pdf}
%    \caption{Common frequency}
%\end{subfigure}
%\begin{subfigure}{0.32\textwidth}
%    \includegraphics[width=\textwidth]{dsdp_innov_freq_local.pdf}
%    \caption{Innovation frequency}
%\end{subfigure}
%\begin{subfigure}{0.32\textwidth}
%    \includegraphics[width=\textwidth]{sdp_freq_local.pdf}
%    \caption{Separate recovery}
%\end{subfigure}
%\caption{Frequency localization from dual polynomial ($J=32$, $n=40$, $m_j=20,\forall j$)}
%\label{figure:frequency:localization}
%\end{figure}
In this section we study the dual problem to check the successful recovery of the optimization \eqref{CA:minimization} \cite{chandrasekaran2007convex}. We establish the conditions the dual certificate should satisfy to guarantee uniqueness and optimality. Denote the optimal solution to \eqref{CA:minimization} by $X^{\star}$, and let $Q = \lbrack q_1^{\ast},\ldots,q_J^{\ast} \rbrack^{\ast}$, where $q_j\in\mathbb{C}^{m_j}$. Then the dual problem of \eqref{CA:minimization} is
\begin{equation}\label{dual_problem}
    \max_Q \langle \Phi^{\ast}Q,X^{\star} \rangle_{\mathbb{R}}\ \st \lVert \Phi^{\ast}Q \rVert_{\mathcal{CA}}^{\ast} \leq 1,
\end{equation}
where $\lVert \cdot \rVert_{\mathcal{CA}}^{\ast}$ is the dual norm of CA-norm, and
\begin{IEEEeqnarray*}{rCl}
    &&\lVert \Phi^{\ast}Q \rVert_{\mathcal{CA}}^{\ast} = \sup_{\lVert X \rVert_{\mathcal{CA}}=1} \langle \Phi^{\ast}Q,X \rangle_{\mathbb{R}}\\
    &=& \sup_{\lVert z_c \rVert_{\mathcal{A}} + \sum_j \lVert z_j \rVert_{\mathcal{A}}=1} \bigg( \big\langle \sum_{j\in\Lambda}\Phi_j^{\ast}q_j,z_c \big\rangle_{\mathbb{R}} + \sum_{j\in\Lambda}\big\langle\Phi_j^{\ast}q_j,z_j \big\rangle_{\mathbb{R}} \bigg)\\
    &=& \sup_{\substack{\lvert c_c \rvert + \sum_j\lvert c_j \rvert = 1 \\ \phi_c,\phi_j \in \left\lbrack 0,2\pi \right),f_c,f_j \in \left\lbrack 0,1 \right\rbrack}} \bigg( \lvert c_c \rvert \big\langle \sum_{j\in\Lambda}\Phi_j^{\ast}q_j,e^{j\phi_c}a(f_c) \big\rangle_{\mathbb{R}}\\
    &&+ \sum_{j\in\Lambda} \lvert c_j \rvert \big\langle\Phi_j^{\ast}q_j,e^{j\phi_j}a(f_j) \big\rangle_{\mathbb{R}} \bigg)\\
    &=& \sup_{f\in\lbrack 0,1 \rbrack} \max \bigg\lbrace \big\lvert \big\langle \sum_{j\in\Lambda}\Phi_j^{\ast}q_j,a(f) \big\rangle \big\rvert, \max_{j\in\Lambda} \big\lvert \big\langle \Phi_j^{\ast}q_j,a(f) \big\rangle \big\rvert \bigg\rbrace.
\end{IEEEeqnarray*}
\indent Strong duality simply holds since \eqref{CA:minimization} is only equality constrained and thus satisfies Slater's condition \cite{boyd2004convex}. Based on this, a dual certificate to the optimality of the solution to \eqref{CA:minimization} can be obtained.
\begin{proposition}\label{proposition:dual_polynomial}
The solution $\hat{X} = X^{\star}$ is the unique optimizer to \eqref{CA:minimization} if there exists a dual polynomial ensemble $Q_j(f) = \langle \Phi_j^{\ast}q_j,a(f) \rangle_{\mathbb{R}},j\in\Lambda$ satisfying
\begin{IEEEeqnarray}{rCl}
    \label{dual:certificate:innovation:realfreq}
    Q_j\left(f_{j,k}\right) &=& \sign\left(c_{j,k}\right),\forall f_{j,k}\in\Omega_j, j\in\Lambda\\
    \label{dual:certificate:common:realfreq}
    \sum_{j\in\Lambda} Q_j\left(f_{c,k}\right) &=& \sign\left(c_{c,k}\right),\forall f_{c,k}\in\Omega_c\\
    \label{dual:certificate:innovation:falsefreq}
    \left\lvert Q_j\left( f \right) \right\rvert & < & 1,\forall f \notin \Omega_j,j\in\Lambda\\
    \label{dual:certificate:common:falsefreq}
    \Big\lvert \sum_{j\in\Lambda} Q_j\left( f \right) \Big\rvert &<& 1,\forall f \notin \Omega_c.
\end{IEEEeqnarray}
\end{proposition}
\begin{proof}
Any $Q$ satisfying the conditions in Proposition \ref{proposition:dual_polynomial} is dual feasible. It also follows that for $X^{\star}$
\begin{IEEEeqnarray*}{rCl}\label{duality:truefreq}
    &&\langle \Phi^{\ast}Q,X^{\star} \rangle_{\mathbb{R}}\\
    &=& \real \Big( \sum_{k=1}^{s_c} c_{c,k}^{\ast} \sum_{j\in\Lambda} \left\langle \Phi_j^{\ast}q_j,a\left( f_{c,k} \right) \right\rangle \Big)\\
    &+& \sum_{j\in\Lambda} \real \Big( \sum_{k=1}^{s_j} c_{j,k}^{\ast} \left\langle \Phi_j^{\ast}q_j,a\left( f_{j,k} \right) \right\rangle \Big)\\
    &=& \sum_{k=1}^{s_c} \left\lvert c_{c,k} \right\rvert + \sum_{j\in\Lambda} \sum_{k=1}^{s_j} \left\lvert c_{j,k} \right\rvert \geq \left\lVert X^{\star} \right\rVert_{\mathcal{CA}},\IEEEyesnumber
\end{IEEEeqnarray*}
where the last inequality is due to the definition of CA-norm. On the other hand, H\"older's inequality \cite{boyd2004convex} states that $\left\langle \Phi^{\ast}Q,X^{\star} \right\rangle_{\mathbb{R}} \leq \left\lVert \Phi^{\ast}Q \right\rVert_{\mathcal{CA}}^{\ast}\left\lVert X^{\star} \right\rVert_{\mathcal{CA}} \leq \left\lVert X^{\star} \right\rVert_{\mathcal{CA}}$, which thus combined with \eqref{duality:truefreq} implies that $\left\langle \Phi^{\ast}Q,Z^{\star} \right\rangle_{\mathbb{R}} = \left\lVert X^{\star} \right\rVert_{\mathcal{CA}}$. Because of strong duality, the primal-dual feasibility of $\left( X^{\star}, \Phi^{\ast}Q \right)$ implies that $X^{\star}$ is a primal optimal solution and $\Phi^{\ast}Q$ is a dual optimal solution \cite{boyd2004convex}.
\newline\indent For uniqueness, suppose $\hat{X}$ with $z_c = \sum_k \hat{c}_{c,k} a(\hat{f}_{c,k})$ and $z_j = \sum_k \hat{c}_{c,j} a(\hat{f}_{c,j})$ is another solution, then we have
\begin{IEEEeqnarray*}{rCl}
    &&\langle \Phi^{\ast}Q,\hat{X} \rangle_{\mathbb{R}}\\
    &=& \sum_{\hat{f}_{c,k}\in\Omega_c}\real \Big( \hat{c}^{\ast}_{c,k} \big\langle \sum_{j\in\Lambda} \Phi_j^{\ast}q_j,a( \hat{f}_{c,k} ) \big\rangle \Big)\\
    &+& \sum_{j\in\Lambda}\sum_{\hat{f}_{j,k}\in\Omega_j} \real \Big( \hat{c}_{j,k}^{\ast} \big\langle \Phi_j^{\ast}q_j, a( \hat{f}_{j,k} ) \big\rangle \Big)\\
    &+& \sum_{\hat{f}_{c,k}\notin\Omega_c} \real \Big( \hat{c}^{\ast}_{c,k} \big\langle \sum_{j\in\Lambda} \Phi_j^{\ast}q_j,a( \hat{f}_{c,k} ) \big\rangle \Big)\\
    &+& \sum_{j\in\Lambda}\sum_{\hat{f}_{j,k}\notin\Omega_j} \real \Big( \hat{c}_{j,k}^{\ast} \big\langle \Phi_j^{\ast}q_j, a( \hat{f}_{j,k} ) \big\rangle \Big)\\
    &<& \sum_{\hat{f}_{c,k}\in\Omega_c} \left\lvert \hat{c}_{c,k} \right\rvert + \sum_{j\in\Lambda}\sum_{\hat{f}_{j,k}\in\Omega_j} \left\lvert \hat{c}_{j,k} \right\rvert\\
    &+& \sum_{\hat{f}_{c,k}\notin\Omega_c} \left\lvert \hat{c}_{c,k} \right\rvert + \sum_{j\in\Lambda}\sum_{\hat{f}_{j,k}\notin\Omega_j} \left\lvert \hat{c}_{j,k} \right\rvert
    = \big\lVert \hat{X} \big\rVert_{\mathcal{CA}}\IEEEyesnumber
\end{IEEEeqnarray*}
due to conditions \eqref{dual:certificate:innovation:falsefreq} and \eqref{dual:certificate:common:falsefreq} if either $z_c$ is not solely supported on $\Omega_c$ or $z_j$ is not solely supported on $\Omega_j$, contradicting strong duality. Therefore, all optimal solutions must have a common component supported on $\Omega_c$ and innovation components supported on $\Omega_j,j\in\Lambda$, respectively. The uniqueness of optimal solution simply follows from the linear independency of the set of atoms with frequencies in $\Omega_c$ or in $\Omega_j$.
\end{proof}
Proposition \ref{proposition:dual_polynomial} serves as a guide for the construction of dual polynomials, of which the details we leave for future work. The construction of dual polynomials following the guide would give a comprehensive theoretical analysis of the performance of the CA-norm minimization. A consequence of this proposition is a way to determine the composing frequencies by evaluating the reconstructed dual polynomial ensemble and identifying the locations where \eqref{dual:certificate:innovation:realfreq} and \eqref{dual:certificate:common:realfreq} hold. An instance of frequency localization is illustrated in Figure \ref{figure:frequency:localization}. For joint recovery, both the dual polynomials and their sum achieve modulus 1 only at truth frequencies, and are strictly bounded in other regions, while the separate recovery suffers from severe inaccuracy and errors.

\section{Numerical Experiments}
\begin{figure}[tb!]
\centering
\includegraphics[width=0.37\textwidth]{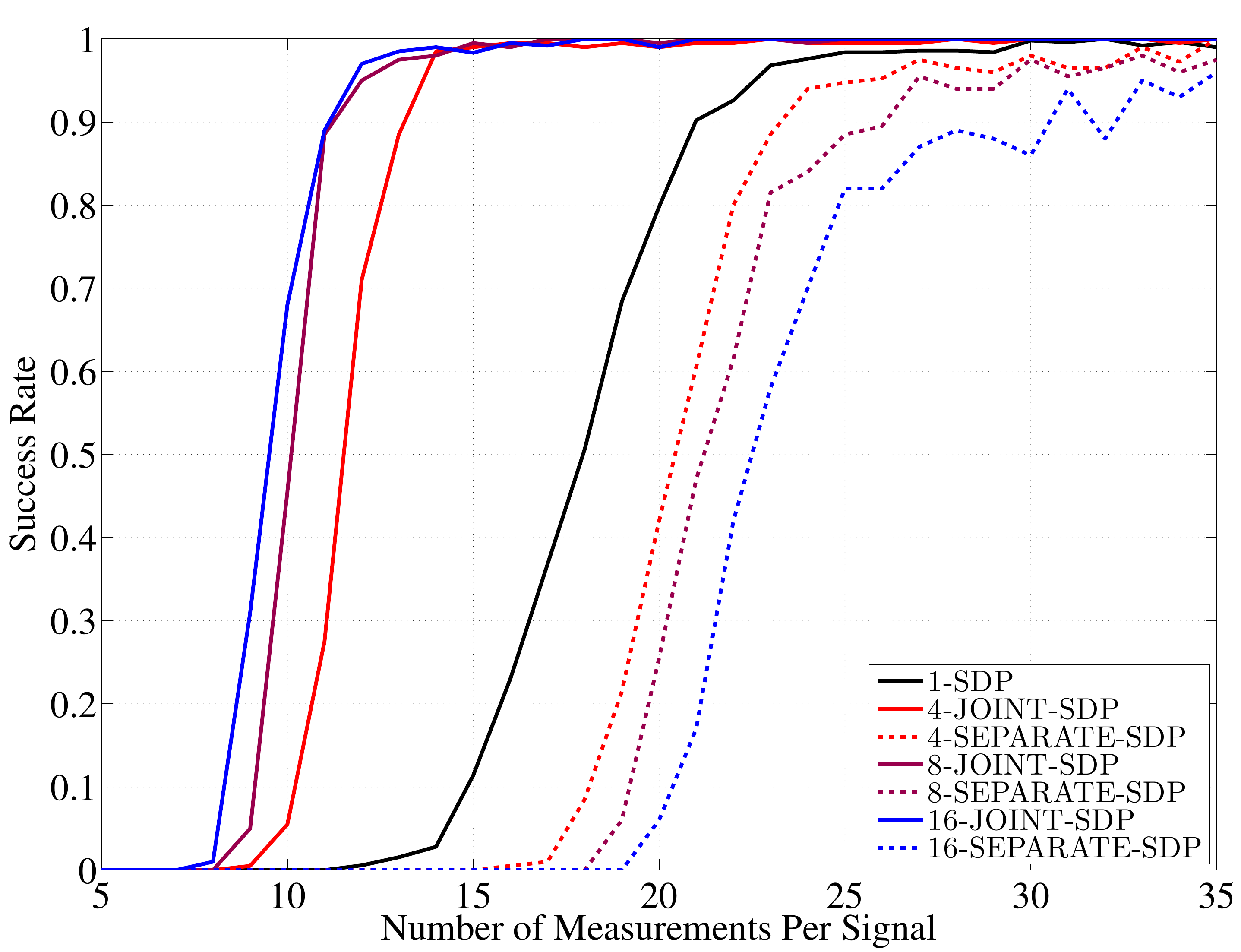}
\caption{Success Rate vs Number of Measurements per Signal}
\label{figure:success_rate}
\end{figure}
In this section, we evaluated the proposed approach by performing numerical experiments. Since the numerical results in our previous work \cite{lu2014spectral} illustrate that the atomic norm minimization yields state-of-the-art performance for noiseless recovery, we simply sidestep other approaches and compare the joint recovery approach to its separate counterpart. We chose the success rate as the major performance measure. The recovery is considered successful if the relative error $\lVert \hat{x}_j - x_j \rVert_2 / \lVert x_j \rVert_2 \leq 10^{-6},\forall j\in\Lambda$ is true. We set $n = 40$, $s_c = 4$, $s_j = 2$ for each signal. Frequencies were generated uniformly random on $\lbrack 0,1 \rbrack$ with an additional constraint on minimum separation $\Delta$ as follows
\begin{equation*}
    \Delta = \min_j \inf_{f,f^{\prime}\in\Omega_c \cup \Omega_j:f \neq f^{\prime}} \lvert f-f^{\prime} \rvert \geq \frac{1}{n}.
\end{equation*}
Phase shifts were selected uniformly random in $\lbrack 0,2\pi )$, and magnitudes were generated as $0.5 + \omega^2$ with $\omega$ a zero-mean unit-variance Gaussian random variable. The sensing matrices were random sub-identity matrices. We performed Monte Carlo experiments for $m_j$ from $5$ to $35$ and $J = 1,4,8,16$, and recorded the success rate from $200$ trials. The joint SDP \eqref{semidefinite:program} was solved via SDPT3-4.0 toolbox \cite{toh1999sdpt}.The performance curves are shown in Figure \ref{figure:success_rate}.\footnote{The authors would like to thank Gongguo Tang for providing the implementation of his algorithm.}
\newline\indent The joint SDP exhibits a definite advantage over its separate counterpart. The joint SDP achieves exact recovery after $m_j$ exceeds a certain threshold. For $J=4$, the intrinsic sparsity of signal ensemble is $K=(4+2\times 4)\times 3 = 36$, since at least three independent parameters are required to determine one sinusoid. The number of measurements required for perfect recovery is $14\times 4 = 56$ for joint SDP, approximately $1.56K$, while separate SDP requires at least $30\times 4 = 120$ measurements to achieve comparable performance, approximately $3.33K$. Hence the joint SDP in practice overcomes the performance bound encountered in separate recovery. The gap increases with the increase of $J$, implying the promise of application to large-scale sensor systems.

\section{Conclusion}
In this letter, we proposed the CA-norm minimization for recovering a JFS signal ensemble sharing a common frequency-sparse component from the collection of their compressed measurements. We established a computationally tractable joint SDP solution to the CA-norm minimization. We also characterized a dual certificate for the optimality of the proposed optimization problem. As shown in Figure \ref{figure:success_rate}, the definite advantage of joint SDP for large $J$ implies the promising application to large-scale sensor systems. The core contribution is twofold. First, we extended off-the-grid formulation to distributed CS framework, providing an instance of addressing signal ensemble with joint structure specified in continuously parameterized dictionaries. Second, the requirements of the certificate polynomials are far more stringent and require a non-trivial modification of construction using additional kernel parts. The successful localization of common frequencies is dependent on the combined contribution of all polynomials.

\balance
\bibliographystyle{IEEEtran}
\bibliography{reference}

% if have a single appendix:
%\appendix[Proof of the Zonklar Equations]
% or
%\appendix  % for no appendix heading
% do not use \section anymore after \appendix, only \section*
% is possibly needed

% use appendices with more than one appendix
% then use \section to start each appendix
% you must declare a \section before using any
% \subsection or using \label (\appendices by itself
% starts a section numbered zero.)
%

%\begin{appendix}
%\subsection{Proof}
%Appendix one text goes here.
%\subsection{Another appendix}
%another appendix
%\end{appendix}

% you can choose not to have a title for an appendix
% if you want by leaving the argument blank
%\section{dd}
%Appendix two text goes here.

% use section* for acknowledgement

% Can use something like this to put references on a page
% by themselves when using endfloat and the captionsoff option.
\ifCLASSOPTIONcaptionsoff
  \newpage
\fi

% trigger a \newpage just before the given reference
% number - used to balance the columns on the last page
% adjust value as needed - may need to be readjusted if
% the document is modified later
%\IEEEtriggeratref{8}
% The "triggered" command can be changed if desired:
%\IEEEtriggercmd{\enlargethispage{-5in}}

% references section

% can use a bibliography generated by BibTeX as a .bbl file
% BibTeX documentation can be easily obtained at:
% http://www.ctan.org/tex-archive/biblio/bibtex/contrib/doc/
% The IEEEtran BibTeX style support page is at:
% http://www.michaelshell.org/tex/ieeetran/bibtex/
%\bibliographystyle{IEEEtran}
% argument is your BibTeX string definitions and bibliography database(s)
%\bibliography{IEEEabrv,../bib/paper}
%
% <OR> manually copy in the resultant .bbl file
% set second argument of \begin to the number of references
% (used to reserve space for the reference number labels box)
%\bibliographystyle{IEEEtran}
%\bibliography{reference}

% that's all folks
\end{document}